\documentclass[letterpaper, 10 pt, conference]{ieeeconf}  

\IEEEoverridecommandlockouts                              

\usepackage{cite}
\usepackage{amsmath} 
\usepackage{amssymb}  
\usepackage{mathrsfs}
\usepackage{amsthm}
\usepackage{bm}  
\usepackage{colortbl}
\usepackage{multirow}
\usepackage{mathtools}
\usepackage[dvipsnames]{xcolor}

\usepackage{multirow}
\usepackage{threeparttable}
\usepackage{mathrsfs}
\usepackage{booktabs}
\usepackage{adjustbox}

\usepackage{algorithmicx}
\usepackage{algorithm}
\usepackage[noend]{algpseudocode}   
\algnewcommand\algorithmicinput{\textbf{Input:}}
\algnewcommand\Input{\item[\algorithmicinput]}

\algnewcommand\algorithmicoutput{\textbf{Output:}}
\algnewcommand\Output{\item[\algorithmicoutput]}

\algnewcommand{\LineComment}[1]{\Statex \(\triangledown \) #1}  

\DeclareMathOperator*{\argmin}{arg\,min}

\usepackage{graphicx}
\graphicspath{{../image/}}
\usepackage{relsize}

\makeatletter
\let\NAT@parse\undefined
\makeatother
\usepackage[hidelinks]{hyperref}
\usepackage{xcolor}
\hypersetup{
	colorlinks,
	linkcolor={red!50!black},
	citecolor={blue!50!black},
	urlcolor={blue!80!black}
}

\newcommand{\cL}{\mathcal L}

\newcommand{\cD}{\mathcal D}

\newcommand{\cA}{\mathcal A}

\newcommand{\bbN}{\mathbb N}
\newcommand{\bbR}{\mathbb R}
\newcommand{\bu}{\bm u}

\newcommand{\DN}{\Delta_N}

\newcommand{\dN}{\delta_N}
\newcommand{\DM}{\Delta_M}

\newcommand{\deq}{\coloneqq}

\newtheorem{definition}{Definition}

\newtheorem{remark}{Remark}
\newtheorem{lemma}{Lemma}
\newtheorem{proposition}{Proposition}

\newtheorem{problem}{Problem}

\begin{document}

\title{
	Optimal Control of Boolean Control Networks with Discounted Cost: An Efficient Approach based on Deterministic Markov Decision Process 
}

\author{Shuhua Gao, Cheng Xiang, and Tong Heng Lee 
	\thanks{All authors are with the Department of Electrical and Computer Engineering, National University of Singapore, Singapore, 117583. Email: {\tt\small elexc@nus.edu.sg}}
}

\maketitle
\thispagestyle{empty}
\pagestyle{empty}

\begin{abstract}
This paper deals with the infinite-horizon optimal control problem for Boolean control networks (BCNs) with a discounted-cost criterion. This problem has been investigated in existing studies with algorithms characterized by high computational complexity. We thus attempt to develop more efficient approaches for this problem from a deterministic Markov decision process (DMDP) perspective. First, we show the eligibility of a DMDP to model the control process of a BCN and the existence of an optimal solution. Next, two approaches are developed to handle the optimal control problem in a DMDP. One approach adopts the well-known value iteration algorithm, and the other resorts to the Madani's algorithm specifically designed for DMDPs. The latter approach can find an exact optimal solution and outperform existing methods in terms of time efficiency, while the former value iteration based approach usually obtains a near-optimal solution much faster than all others. The 9-state-4-input \textit{ara} operon network of the bacteria \textit{E. coli} is used to verify the effectiveness and performance of our approaches. Results show that both approaches can reduce the running time dramatically by several orders of magnitude compared with existing work.

\end{abstract}

\section{Introduction}
An effective and widely used model of gene regulatory networks \cite{barabasi2004network} is the Boolean network (BN) model, first proposed by Kauffman in 1969 \cite{kauffman1969metabolic}, that describes gene expression state with binary values. Since then, BNs have drawn a lot of research interest and been applied to various fields beyond biomolecular networks, such as information mining in consumer community networks \cite{meng2018properties} and analysis of social consensus impacted by peer interactions \cite{Emergence2007}. We can further incorporate binary control inputs into a BN to manipulate its states and get a control system commonly referred to as a \textit{Boolean control network} (BCN) \cite{zhao2010input}. 

A considerable number of studies on BCNs emerged in the last decade thanks to the development of a novel mathematical tool called the semi-tensor product (STP) \cite{zhao2010input, cheng2010linear}. An equivalent algebraic state-space representation (ASSR) can be built using STP, which makes it possible to adapt established techniques in traditional control theory for similar investigations of BCNs. Based on the STP and the ASSR of BCNs, quite a few control-theoretical problems have been tackled in the recent literature, for example, controllability and observability \cite{cheng2009controllability, laschov2013observability, zhao2010input}, stabilization \cite{cheng2011stability}, pinning control \cite{lu2019pinning},  and output tracking \cite{zhang2019output}, to name a few. Following this mainstream, we also initiate our study on infinite-horizon optimal control of BCNs with the ASSR here.

Optimal control is a classic topic that deals with the design of an \textit{optimal} control law according to a given performance index. Specifically,  optimal control of BCNs can be used to develop medical intervention strategies for an underlying GRN to treat diseases like cancers while minimizing expenses or maximizing the therapeutic effect \cite{faryabi2008optimal}. A variety of optimal control problems regarding BCNs have been studied in recent years, which are divided into two broad categories depending on the optimization horizon length. In the first class, the horizon length is finite, and the performance criterion is the summation of stage costs at a countable number of time steps as well as one terminal cost. An early study was conducted in \cite{laschov2010maximum} towards the Mayer-type optimal control (i.e., only considering the terminal cost) of single-input BCNs by a maximum principle. Two common objectives in optimal control, minimum energy, and minimum time, have been attempted in \cite{li2013minimum} and \cite{laschov2013minimum}, respectively. E. Fornasini \textit{et al.} investigate more general cases of such finite-horizon problems in \cite{fornasini2013optimal} and present recursive algorithms that are analogous to the discrete-time Riccati equation. The second class of problems, i.e., infinite-horizon optimal control, are generally more challenging, of which the objective function takes either an average-cost form or a discounted-cost form to ensure the convergence of the total cost \cite{faryabi2008optimal}. The first attempt for infinite-horizon optimal control with an average-cost criterion was presented in \cite{zhao2010optimal} by enumerating all cycles in the input-state space with prohibitively high time complexity. Several improvements were proposed later, including a Floyd-like algorithm \cite{zhao2011floyd}, a value iteration algorithm \cite{fornasini2013optimal}, and a policy iteration approach \cite{wu2019optimal}. By contrast, the discounted-cost counterpart has got less attention, which was first addressed in \cite{cheng2014optimal} using a Floyd-like algorithm similar to that in \cite{zhao2011floyd}. The algorithm \cite{zhao2011floyd} has been modified in a recent study \cite{zhu2018optimal} to operate in the state space instead of the input-state space of a BCN for further speedup.

A major issue of the STP-based algebraic methods discussed above is their prohibitively high computational cost once the size of the BCN is large. It has been proved in \cite{AKUTSU2007670} that, in general, control problems on BCNs are NP-hard. Consequently, it is hopeless to seek polynomial-time algorithms since P $ \ne $ NP is a widely believed conjecture. This is indeed an intuitive fact because all algorithms above run in a polynomial time of $ N $, where $ N \deq 2^n$ and $ n $ is the number of state variables in a BCN. Nevertheless, even faced with the NP-hardness, we can still pursue shorter running time in practice by designing algorithms whose time complexity is a lower-order polynomial in $ N $. For example, by resorting to the Warshall algorithm, Liang \textit{et al. }\cite{liang2017improved} proposed an improved controllability criterion for BCNs with time complexity reduced from $ O(N^4) $ to $ O(N^3) $. Our latest work \cite{gao2019infinite} (preprint) investigates infinite-horizon optimal control of BCNs with average cost using Karp's minimum mean cycle (MMC) algorithm and achieves the lowest time complexity so far. Notably, regarding the discounted-cost optimal control problem considered in this paper, the existing two studies \cite{cheng2014optimal} and \cite{zhu2018optimal} both attempt to locate the overall optimal cycle by examining individual optimal cycles of length ranging from 1 to $ N $ iteratively, which consequently leaves ample space for further efficiency improvement. 

The primary goal of this study is to develop more efficient algorithms for discounted-cost infinite-horizon optimal control of BCNs. As a natural choice, the Markov decision process (MDP) theory has been extensively used in optimal control of probabilistic and stochastic Boolean networks, e.g., see \cite{faryabi2008optimal} and \cite{wu2017finite}. Though a deterministic BCN considered here can undoubtedly be treated as a special stochastic BCN, more complexity will be introduced that causes unnecessary deterioration of computational efficiency. To the best of our knowledge, there is currently no work on optimal control of BCNs that views the control process as a deterministic Markov decision process (DMDP).  The interesting point is that, by adopting the equivalent DMDP description, we can resort to established algorithms, like Madani's algorithm \cite{madani2010discounted}, to solve the discounted-cost optimal control problem for BCNs with reduced time complexity. The development of such efficient, DMDP-based algorithms forms the main contribution of this paper.

The rest of this paper is organized as follows. First, in Section \ref{sec: preliminary}, we introduce the algebraic representation of BCNs. We then formulate the optimal control problem in Section \ref{sec: problem}. The main results of our study are presented in Section \ref{sec: results}, which detail the development of two efficient approaches. We compare the performance of the proposed approaches and existing ones on a biological network  in Section \ref{sec: example}. Finally, Section \ref{sec: conclusion} concludes this study.
The Python implementation of all algorithms in this paper is available at \url{https://github.com/ShuhuaGao/bcn_opt_dc}.

\section{Preliminaries} \label{sec: preliminary}
\subsection{Notations}
\begin{itemize}
	\item $ \bbR $, $ \bbN $, and $ \bbN^+ $ denote the sets of real numbers, nonnegative integers, and positive integers, respectively. Given $ k , n \in \bbN $ with $ k \le n $, $ [k, n] \deq \{k, k + 1, \cdots, n  \}$.
	\item $\mathcal{A}_{p\times q}$ denotes the set of all $p\times q$ matrices. Given $ A \in \cA $, $ A_{ij} $ is its $ (i, j) $-th entry, and $ \textrm{Row}_i(A) $, $\textrm{Col}_j(A) $ denote its $ i $-th row and $ j $-th column respectively. 
	\item $ \delta_n^i \deq \textrm{Col}_i(I_n) $, where $ I_n $ is the $ n $-dimensional identity matrix. $\Delta_n \deq \{ \delta_n^i | i = 1, 2, \cdots, n \}$, and $ \Delta \deq \Delta_2 $.  The shorthand of $ \{\delta_n^{i_1}, \delta_n^{i_2}, \cdots, \delta_n^{i_k}\} $ is $ \delta_n\{i_1, i_2, \cdots, i_k\} $.
	\item A matrix $ L \in \cA_{n \times q} $ with $ \text{Col}_i(L) \in \Delta_{n}, \forall i \in [1, q], $ is called a \textit{logical matrix}. Let $\mathcal{L}_{n\times q} $ denote the set of all $ n\times q $ logical matrices.
	\item $ \cD \deq \{0, 1\}. $ Logical operators \cite{cheng2010linear}:  $\land$, conjunction; $\lor$, disjunction; $ \lnot $, negation; and $ \oplus $, exclusive or.
\end{itemize}

\subsection{Algebraic Representation of BCNs}
\begin{definition}\cite{zhao2010optimal}
	The semi-tensor product (STP) of two matrices $ A \in   \mathcal{M}_{m\times n}$ and $ B \in   \mathcal{M}_{p\times q}$ is defined by
	\begin{equation*}
	A \ltimes B = (A \otimes I_{\frac{s}{n}})(B \otimes I_{\frac{s}{p}}),
	\end{equation*}
	where $\otimes$ denotes the Kronecker product,  and $ s $ is the least common multiple of $ n $ and $ p $. $ \ltimes_{i=1}^n A_i \deq A_1\ltimes A_2\ltimes \cdots\ltimes A_n $.
\end{definition}
\begin{remark}
	The STP generalizes the traditional matrix product while preseving most fundamental properties \cite{cheng2010linear}. For notational simplicity, the symbol $\ltimes$ is omitted hereafter.
\end{remark}

Identify Boolean values in $ \cD $ by $ 0 \sim \delta_2^1 $ and $ 1 \sim \delta_2^2 $. 
\begin{lemma} \cite{cheng2010linear} \label{lemma: structure matrix}
	Any Boolean function $ f(x_1, x_2, \cdots, x_n): \Delta^n \rightarrow \Delta $ can be expressed uniquely in a multi-linear form as 
	\begin{equation}
	f(x_1, x_2, \cdots, x_n) = M_f x_1 x_2 \cdots x_n,
	\end{equation}
	where $ M_f \in  \mathcal{L}_{2\times 2^n}$ is the unique \textit{structure matrix} of $ f $.
\end{lemma}

Consider a BCN with $ n $ nodes and $m  $ control inputs:
\begin{equation} \label{eq: bcn}
\begin{cases}
x_1(t+1) = f_1(x_1(t), \cdots, x_n(t), u_1(t), \cdots, u_m(t))\\
\vdots \\
x_n(t+1) = f_n(x_1(t), \cdots, x_n(t), u_1(t), \cdots, u_m(t)),
\end{cases}
\end{equation}
where $ x_i(t) \in \Delta, u_j(t) \in \Delta, $ denote states and control inputs respectively, and $ f_i : \Delta^{m+n} \rightarrow \Delta $ is the Boolean function associated with the state variable $ x_i $, $ i \in [1, n], j \in [1, m] $. 

Using the STP, the ASSR of the BCN \eqref{eq: bcn} is 
\begin{equation} \label{eq: ASSR}
x(t+1) = Lu(t)x(t),
\end{equation}
where $  x(t) \deq x_1(t) \ltimes \cdots \ltimes x_n(t) \in \Delta_{2^n} $ and $  u(t) \deq u_1(t) \ltimes \cdots \ltimes u_m(t) \in \Delta_{2^m} $ are canonical vectors. Let $ N \deq 2^n $ and $ M \deq 2^m $. We have the logical matrix $ L \in \cL_{N \times MN} $. Ref. \cite{cheng2010linear} details the computation of \eqref{eq: ASSR}. Note that the two notations $ N $ and $ M $ defined here are used throughout this paper.

\section{Problem Formulation} \label{sec: problem}
Given the BCN \eqref{eq: ASSR}, let the cost of applying control $ u \in \DM $ at state $ x \in \DN $ be $ g(x, u) $. The bounded function $ g: \DN \times \DM \rightarrow  \bbR$ is called the \textit{stage cost} function. We seek a control sequence that minimizes the discounted cost for BCN \eqref{eq: ASSR} accumulated  in an infinite horizon. Note that we consider a more general and  challenging scenario here beyond that in \cite{cheng2014optimal} and \cite{zhu2018optimal}, which involves various constraints on both states and inputs. The problem is formalized as follows.
\begin{problem} \label{prob: 1}
	Consider BCN \eqref{eq: ASSR}. Solve the following constrained optimization problem for optimal control:
	\begin{align} \label{eq: problem}
	\min_{\bu} J(\bu) = \lim_{T\rightarrow\infty}  \sum_{t=0}^{T-1} \lambda^tg(x(t), u(t)), \nonumber\\
	\textrm{s.t.} \begin{cases}
	x(t+1) = Lu(t)x(t) \\
	x(t) \in C_{\textnormal{\textrm{x}}} \\
	u(t) \in C_{\textnormal{\textrm{u}}}(x(t))  \\
	x(0) = x_0 
	\end{cases},
	\end{align}
	where $ \bu = \big(u(t) \in \Delta_M\big)_{t=0}^{T-1} $ denotes a control sequence; $ \lambda \in (0, 1) $ is the discount factor; $ C_{\textnormal{\textrm{x}}} \subseteq \DN$ and $ C_{\textnormal{\textrm{u}}}(x(t)) \subseteq \DM$ denote the state constraints and the state-dependent control input constraints respectively; and  $ x_0 \in C_{\textnormal{\textrm{x}}}$ is the initial state of the BCN. 
\end{problem}

\begin{remark} \label{rmk: problem}
	No constraints are considered in \cite{cheng2014optimal}, and only the avoidance of undesirable states is handled in \cite{zhu2018optimal}. By contrast, the above problem formulation emerges as the most generic one, which can incorporate state constraints, control constraints, and transition constraints \cite{zhang2017finite}. We assume that Problem \ref{prob: 1} is feasible, that is, at least one control sequence exists that allows the indefinite evolution of the BCN.
\end{remark}

\section{Main Results} \label{sec: results}
In this section, we first show that the control of a BCN can be handled elegantly in an MDP framework. Then, we propose two methods to solve Problem \ref{prob: 1}: a general value iteration approach commonly used in MDP optimization and a more efficient approach specialized for a DMDP.

\subsection{Deterministic Markov Decision Process (DMDP)}
An MDP is a widely used mathematical model in sequential decision making under uncertaintis, that is, choosing differente actions in different situations \cite{RL2}. Specificially, in our application with BCNs, the \textit{action} at time point $ t $ refers to the control input $ u(t) $, and the \textit{situation} is represented by the network state $ x(t) $. In the MDP framework, each decision is associated with a \textit{reward}. The essential property of an MDP is that the next state and the reward depend only on the current state and the current action, known as the \textit{Markov property} \cite{RL2}. Obviously, we see from \eqref{eq: ASSR} that the control process of a BCN is indeed an MDP, because $ x(t+1) $ is completed determined by $ x(t) $ and $ u(t) $.

In an MDP, the goal of the controller is to maximize the cumulative reward from any initial state in the long run \cite{RL2,busoniu2017reinforcement}. A \textit{policy} is a decision rule that specifies which action should be chosen for each state. In our BCN application, the \textit{reward} is replaced by the \textit{cost} in Problem \ref{prob: 1}. Accordingly, we aim to find a policy that minimizes the aggregated discounted cost over the infinite horizon for optimal control of BCNs.

Unlike  general MDPs considered in reinforcement learning, a useful property of the BCN control process is that its state transition and rewarding are both deterministic. That is, given the current state $ x \in \DN $ and the control action $ u \in \DM $, the next state is definitely  $ Lux $ by \eqref{eq: ASSR}, and the cost is fixed to $ g(x, u) $ in Problem \ref{prob: 1}. Formally, the control process of a BCN is called a deterministic Markov decision process (DMDP). As we will show later, such determinism allows the development of time-bounded optimization algorithms compared with those for general MDPs.

\subsection{Existence of Optimal Solutions}
In control of BCNs, a policy $ \pi $ refers to a mapping from states to control inputs, i.e., $ \pi: \DN \rightarrow \DM $. A feasible policy must respect the constraints of Problem \ref{prob: 1}: for any $ x \in C_{\textrm{x}} $, it must satisfy
\begin{equation}
\pi(x) \in C_{\textrm{u}}(x), \ L\pi(x)x \in C_{\textrm{x}}.
\end{equation}

Now we can restate Problem \ref{prob: 1} using the MDP terminology as follows: find an optimal policy $\pi_* $, which conforms to all constraints, such that the performance index function $ J $ is minimized. The first question coming to our mind is whether an optimal policy exists for Problem \ref{prob: 1}. In the following illustration, we mainly borrow the notations and terminology from the monograph \cite{RL2}. Note that we are dealing with a DMDP, and all probabilistic expectations in the general MDP framework can thereby be omitted.

The quality of a policy can be evaluated by a value function \cite{busoniu2017reinforcement}.
Given a policy $ \pi $, the\textit{ value function} of a state $ x $, termed $ v_{\pi}(x) $, is the performance index obtained with the initial state $ x $  and the control sequence $ \bu $ generated by $ \pi $:
\begin{equation} \label{eq: value function}
v_{\pi}(x) = \sum_{t=0}^{\infty} \lambda^tg(x(t), \pi(x(t))) \bigg|_{x(0) = x}, \ x \in C_{\textrm{x}}.
\end{equation}
For simplicity, we set $ v_{\pi}(x) = \infty $ for $ x \notin C_{\textrm{x}} $. Let the next state be $ x' = L\pi(x)x $. From \eqref{eq: value function}, the recursion below holds
\begin{equation}\label{eq: vpi}
v_{\pi}(x) = g(x, \pi(x)) + \lambda v_{\pi}(x'), \ x \in C_{\textrm{x}}.
\end{equation}

Since we aim to minimize the  cost, we say a policy $ \pi $ is better than another policy $ \pi' $ if and only if $ v_{\pi}(x) \le v_{\pi'}(x), \forall x \in C_{\textrm{x}} $. The optimal value function $ v_* $ and the optimal policy $ \pi_* $ are specified by
\begin{align}
v_*(x) &= \min_{\pi}v_{\pi}(x), \label{eq: v*}\\
\pi_*(x) &= \argmin_{u \in C_{\textrm{u}}(x)} g(x, u) + \lambda v_*(Lux). \label{eq: pi*}
\end{align}
Further, there holds obviously $ v_{\pi_*}(x) = v_*(x) , \forall x \in C_{\textrm{x}},$ by the Bellman optimality equation \cite{RL2,busoniu2017reinforcement}, given below
\begin{equation} \label{eq: Bellman}
v_{*}(x) = \min_{u \in C_{\textrm{u}}(x)} g(x, u) + \lambda v_*(Lux), \ x \in C_{\textrm{x}},
\end{equation}

A fundamental result in the MDP theory is that the infinite sum in \eqref{eq: value function} has a finite value as long as the reward sequence is bounded \cite{RL2}. As aforementioned in Section \ref{sec: problem}, it is natural and common to set up a bounded stage cost function $ g $ \cite{cheng2014optimal,zhu2018optimal,zhao2010optimal,wu2019optimal}, which implies a finite value function \eqref{eq: ASSR} for each state. Additionally, recall that the number of states and the number of control inputs are both finite in BCN \eqref{eq: ASSR}, i.e., $ N $ and $ M $, respectively. Consequently, the number of possible policies in our case is also finite, which is at most $ M^N $ after constraint-violating ones are eliminated. Note that we assume Problem \ref{prob: 1} is feasible, i.e., at least one policy exists that violates no constraints (see Remark \ref{rmk: problem}). By the policy improvement theorem \cite{RL2}, an optimal policy always exists that minimizes the value function for all states, from which we can construct the optimal  control sequence for Problem \ref{prob: 1} (see Section \ref{sec: vi}). The correctness of the following proposition is obvious.

\begin{proposition}
	Consider Problem \ref{prob: 1}. There exists an optimal control sequence if the stage cost function $ g $ is bounded.
\end{proposition}

\begin{remark}
	The existence of solutions to infinite-horizon optimal control of BCNs with discounted cost (no constraints involved) has been shown in \cite{cheng2014optimal} and \cite{zhu2018optimal} from other aspects instead of the DMDP here. Note that the optimal control strategies for Problem \ref{prob: 1} may not be unique.
\end{remark}

\subsection{Value Iteration based Approach} \label{sec: vi}
A widely used method in searching optimal policies for finite MDPs is \textit{value iteration}, a dynamic programming based algorithm, which attempts to estimate the optimal value function \eqref{eq: value function} of each state via iterative update \cite{RL2, busoniu2017reinforcement}. It is intuitive to derive the update rule in value iteration from the Bellman optimality equation. Recall that the BCN control process is essentially a DMDP, and its optimality equation has been presented in \eqref{eq: Bellman}.

Given an initial guess of the value function, termed $ V(\cdot) $, value iteration works by updating the value function following a rule similar to the optimality equation \eqref{eq: Bellman}:
\begin{equation} \label{eq: Bellman update}
V(x) = \min_{u \in C_{\textrm{u}}(x)} g(x, u) + \lambda V(Lux), \ x \in C_{\textrm{x}}.
\end{equation}
Such update is repeated iteratively until the value function converges for all states, i.e., the change between two iterations gets small enough below a threshold $ \theta \ge 0$.  After the update loop is terminated, we can determine an (approximate) optimal policy from the value function by
\begin{equation} \label{eq: optimal policy}
	\pi_*(x) = \argmin_{u \in C_{\textrm{u}}(x)} g(x, u) + \lambda V(Lux), \ x \in C_{\textrm{x}}.
\end{equation}

Next, a state feedback control law for optimal control can be directly constructed from the optimal policy \eqref{eq: optimal policy} with the following proposition.
\begin{proposition} \label{prop: feedback matrix}
	Consider Problem \ref{prob: 1}. If $ \pi_* $ is an optimal policy for the associated discounted-cost DMDP, then infinite-horizon optimal control can be achieved by stationary state feedback $ u=Kx $, where nontrivial columns of the matrix $ K \in \cL_{M \times N} $ are specified by 
	\begin{equation} \label{eq: state feedback}
		\textnormal{\textrm{Col}}_i(K) = \pi_*(\dN^i),  \ \textrm{if } \dN^i \in C_{\textrm{x}},
	\end{equation}
	with the other columns arbitrarily set.
\end{proposition}

\begin{proof}
	Note that states and control inputs of BCN \eqref{eq: ASSR} are both logical vectors filled with all zeros except a single entry of value 1. We thus have $ K\dN^i = \textrm{Col}_i(K) =  \pi_*(\dN^i)$ for any $ \dN^i \in C_{\textrm{x}} $. That is, we are exactly taking the optimal policy by applying the state feedback law \eqref{eq: state feedback}. By the definitions in \eqref{eq: v*} and \eqref{eq: pi*}, the optimal policy minimizes the value function for each state $ x \in C_{\textrm{x}} $, and $ v_*(x_0) $ is therefore the minimum of the performance index $ J(\cdot) $. 
\end{proof}

The value iteration routine for Problem \ref{prob: 1} is listed in Algorithm \ref{alg: value iteration}. In practice, a small positive threshold $ \theta > 0 $ is used to acquire a sub-optimal solution with an affordable computational cost, since this algorithm generally cannot converge to the exact optimimum in a finite number of iterations \cite{RL2, busoniu2017reinforcement}. Supposing there are $ P $ iterations required for a specific $ \theta $, the computational cost of the loop (Line \ref{line: repeat} - \ref{line: until} ) is $ O(PMN) $. The computation of \eqref{eq: optimal policy} and \eqref{eq: state feedback} runs in $ O(MN) $ and $ O(N) $ respectively. In summary, the time complexity of Algorithm \ref{alg: value iteration} is $ O(PMN) $. Finally, we note that the state feedback controller \eqref{eq: state feedback} is independent of the initial state $ x_0 $. Given an initial state $ x_0 $, the optimal control sequence can be computed readily from \eqref{eq: state feedback} by evolving the BCN from state $ x_0 $ with the control law \eqref{eq: state feedback}.

\begin{algorithm}[htb]
	\caption{Optimal control based on value iteration} \label{alg: value iteration}
	\begin{algorithmic}[1] 
		\Input Problem \ref{prob: 1}: $ L, C_{\textrm{u}}(\cdot), C_{\textrm{x}}, \lambda $. Threshold $ \theta \ge 0 $.
		\Output Optimal state feedback matrix $ K $
		\State Initialize the value function $ V(x) $ arbitrarily for $ x \in C_{\textrm{x}} $
		\Repeat \label{line: repeat}
			\State $ \psi \gets 0 $
			\ForAll{$ x \in C_{\textrm{x}} $}
				\State $ v \leftarrow V(x) $
				\State Update $ V(x) $ by \eqref{eq: Bellman update}
				\State $ \psi \gets \max(\psi, | v - V(x)|) $
			\EndFor
		\Until{$\psi < \theta$} \label{line: until}
		\State Resolve the optimal policy $ \pi^* $ by \eqref{eq: optimal policy}
		\State Construct the matrix $ K $ by Proposition \ref{prop: feedback matrix}
	\end{algorithmic}
\end{algorithm}

\subsection{Madani's Algorithm based Approach}
The primary drawback of the basic value iteration approach in Algorithm \ref{alg: value iteration} is that the number of iterations to get the exact optimal control strategy is not bounded \cite{RL2,busoniu2017reinforcement}. Consequently, only a sub-optimal solution can be acquired in practice. On the other hand, recall that value iteration is a general algorithm for MDPs, especially stochastic ones, while our BCN control is more precisely a DMDP. In \cite{madani2010discounted}, exploiting the determinism of a DMDP, Madani \textit{et al.} develops a specialized and more efficient algorithm for solving discounted-cost DMDP problems. A more desirable advantage of this algorithm is its guarantee that exact solutions can be obtained in finite steps. In this section, we develop a more efficient and effective method to solve Problem \ref{prob: 1} by resorting to Madani's algorithm \cite{madani2010discounted}.

Madani's algorithm handles discounted-cost DMDPs from a graphical perspective and can be viewed as an adaptation of Karp's algorithm for average-cost DMDPs \cite{gao2019infinite}. In the context of optimal BCN control, the DMDP is described by the state transition graph (STG) of the BCN, termed $ G = (V, E) $, where each vertex represents a state, i.e., $ V \deq  C_{\textrm{x}}$, and each edge denotes a state transition, i.e., 
\begin{equation}
E = \{(x, x') \in C_{\textrm{x}} \times C_{\textrm{x}}  |  \exists u \in C_{\textrm{u}}(x), x' = Lux  \}.
\end{equation}

The weight of each edge is the minimal cost of the corresponding state transition, since a transition may be attained by more than one control input at different costs. Consider two connected states (vertices) in $ G $, say $ (x, x') \in E $. The set of admissible control inputs for this transition (edge) is 
\begin{equation}
U_{xx'} = \{ u \in C_{\textrm{u}}(x) | x' = Lux \},
\end{equation}
and the weight of this edge is 
\begin{equation} \label{eq: w}
w(x, x') = \min_{u \in U_{xx'}} g(x, u),
\end{equation}
along with the best control input enabling this transition
\begin{equation} \label{eq: u*}
u^*(x, x') = \argmin_{u \in U_{xx'}} g(x, u).
\end{equation}
Note that the best control input in \eqref{eq: u*} may not be unique, and we can choose an arbitrary one in that case. Besides, the technique by \eqref{eq: w} and \eqref{eq: u*} can also be adapted to the above value iteration approach to first filter out unlikely actions for specific states to improve computational efficiency.

Given BCN \eqref{eq: ASSR} with constraints in Problem \ref{prob: 1}, it is easy to construct the STG  $ G $ following a breadth-first search (BFS) routine, whose details can be found in our previous work \cite{gao2019infinite}. After the STG is available, Madani's algorithm works in three stages, like follows.
\begin{enumerate}
	\item Compute the minimal discounted cost of a $ k $-edge path starting from each vertex $ x \in C_{\textrm{x}} $, termed $ d_k(x) $, for each $ k \in [1, |C_{\textrm{x}}| ]$ with $ d_0(x) = 0 $.
	\item Compute the quantity below for each vertex $ x \in  C_{\textrm{x}}$:
	\begin{equation} \label{eq: y0}
	y_0(x) = \max_{0 \le k < |C_{\textrm{x}}| } \frac{d_{|C_{\textrm{x}}|}(x) - \lambda^{|C_{\textrm{x}}| - k}d_k(x)}{1 - \lambda^{|C_{\textrm{x}}| - k}}.
	\end{equation}
	\item Recompute the the minimal discounted cost of a $ k $-edge path from each vertex $ x \in C_{\textrm{x}} $, termed $ y_k(x) $, but with the initial value $ y_0(x) $ in \eqref{eq: y0}, for $ 1 \le k < |C_{\textrm{x}}|  $.
	\item The optimal value function of each state (vertex) $ x \in C_{\textrm{x}} $ is obtained by
	\begin{equation} \label{eq: v*2}
	v_*(x) = \min_{0 \le k < |C_{\textrm{x}}| } y_k(x).
	\end{equation}
\end{enumerate}

Interested readers can refer to \cite{madani2010discounted} for detailed proof of the correctness of this algorithm. In practical implementation, the above tasks 1) and 3) can be done efficiently via dynamic programming in a form like Bellman optimality equation \eqref{eq: Bellman}. The corresponding pseudocode is presented in Algorithm \ref{alg: Madani}. Once the optimal value function $ v_* $ is obtained, we can again, just like Algorithm \ref{alg: value iteration}, get the optimal policy by \eqref{eq: optimal policy} and the optimal state feedback law by Proposition \ref{prop: feedback matrix}. 

\begin{algorithm}[htb]
	\caption{Optimal control based on Madani's algorithm} \label{alg: Madani}
	\begin{algorithmic}[1] 
		\Input Problem \ref{prob: 1}: $ L, C_{\textrm{u}}(\cdot), C_{\textrm{x}}, \lambda $.
		\Output Optimal state feedback matrix $ K $
		\State Build the STG $ G = (V, E) $ (see \cite{gao2019infinite} for details)
		\State $ d_0(x) \gets 0 $ for each $ x \in  V$
		\ForAll{$ k \in [1, |V|] $}
			\ForAll{$ x \in  V$}
				\State $ d_k(x) \gets \min_{(x, x') \in E}  w(x, x') + \lambda d_{k-1}(x')$
			\EndFor
		\EndFor
		\ForAll{$ x \in V $}
			\State Compute $ y_0(x) $ by \eqref{eq: y0}
		\EndFor
		\ForAll{$ k \in [1, |V| - 1] $}
			\ForAll{$ x \in  V$}
				\State $ y_k(x) \gets \min_{(x, x') \in E}  w(x, x') + \lambda y_{k-1}(x')$
			\EndFor
		\EndFor
		\ForAll{$ x \in V $}
			\State Compute $ v^*(x) $ by \eqref{eq: v*2}
		\EndFor
		\State Get the optimal policy $ \pi^* $ by \eqref{eq: optimal policy}
		\State Construct the matrix $ K $ by Proposition \ref{prop: feedback matrix}
	\end{algorithmic}
\end{algorithm}

As we have analyzed in \cite{gao2019infinite}, the time complexity to build the STG $ G = (V, E) $  subject to constraints in Problem \ref{prob: 1} is $ O(MN) $. The running time of Madani's algorithm in the graph $ G $ is $ O(|V||E|) $ \cite{madani2010discounted}. Note that there are at most $ N $ vertices in the STG, i.e., $ |V| \le N $, and each vertex has at most $ M $ outgoing edges, which means $ |E| \le M|V| \le MN $. Therefore, the running time of Algorithm \ref{alg: Madani} is dominated by the Madani's part, which is consequently $ O(MN^2) $.

\section{A Biological Example: \textit{Ara} Operon Network } \label{sec: example}
In this section, we apply the two approaches proposed above to the \textit{ara} operon network in the bacteria \textit{E.coli} and compare its performance with that of existing methods. The ara operon network is a well studied GRN that plays a key role in  metablism of  the sugar \textit{L-arabinose} in the absence of glucose.  The GRN's BCN model has 9 state variables (nodes), listed in Table \ref{tbl:ara}, and 4 control inputs, $ A_e $, $ A_{em} $, $ A_{ra\_} $, and $ G_e $. The Boolean functions associated with each node are also listed in Table \ref{tbl:ara}. More biological knowledge of this network is available in \cite{jenkins2017bistability}. Its ASSR \eqref{eq: ASSR} has a structure matrix $ L \in \cL_{512 \times 8192} $ with $ M = 16 $ and $ N = 512 $, which is presented in the online material.

\begin{table}[tb]
	\centering
	\caption{BCN model of the \textit{ara} operon network}
	\label{tbl:ara}
	\begin{tabular}{@{}ll|ll@{}}
		\toprule
		Node     & Function                                       & Node & Function \\ \midrule
		$ A $     & $ A_e \land T $                           & $ D $ &  $\lnot A _{ra_+}  $ $ \land $ $A _{ra_-}  $  \\
		$ A_m $ & $ (A_{em} \land T) \lor A_e$ & $ M_S $ & $ A _{ra_+} \land C \land \lnot D $     \\
		$ A _{ra_+} $&   $ (A_m \lor A) \land A _{ra_-} $        &$ M_T $    &    $ A _{ra_+}\land C $        \\
		$C  $&       $ \lnot G_e $B                                         &$ T $     &$ M_T $         \\
		$ E $& $ M_S $                                                &     &          \\
		\bottomrule
	\end{tabular}
\end{table}

 Wu\textit{ et al.} have investigated the infinite-horizon optimal control of the \textit{ara} operon network with average cost in \cite{wu2019optimal}. We reuse their stage cost function in this study as follows:
 \begin{equation}
 g(x, u) = AX + BU
 \end{equation}
 with the column vectors $ X= [x_1, x_2, \cdots, x_9]^{\top}, U = [u_1, u_2, u_3, u_4]^{\top}$ and the two weight vectors as 
 \begin{equation*}
 A = [-28, -12, 12, 16, 0, 0, 0, 20, 16], \ B = [-8, 40, 20, 40].
 \end{equation*}
 We assume an initial state $ x_0 = \delta_{512}^{10} $ and a discount factor $ \lambda = 0.5 $. No constraints are applied here for comparison purpose, since existing methods are not designed to handle constraints. In the value iteraton approach, the $ \epsilon $-suboptimal solutions are obtained. We implement all algorithms in Python 3.7 and measure their running time for Problem \ref{prob: 1} on a laptop PC with a 1.8 GHz Core i7-8550U CPU, 8 GB RAM, and 64-bit Windows 10.  All methods obtain the same optimal value, $ J^* = 5.232 $, except that the value iteration approach gets an approximate one. 
 
 We gather the theoretical time complexity and the measured running time of each method in Table \ref{tbl: comparison}. As we see, the huge difference in running time between different methods accords well with previous time complexity analysis. Clearly, the two DMDP based approaches proposed in this paper can significantly reduce the running time. Note that, though Algorithm \ref{alg: value iteration} has no upper bound on the number of iterations to get an exact optimum, it usually converges very fast in practice if only a suboptimal solution is desired. For example, only 9, 13, and 18 iterations are needed in this case for the three thresholds in Table \ref{tbl: comparison}. Overall, the take-home message is that one can first try Algorithm \ref{alg: value iteration} based on value iteration and then resorts to Algorithm \ref{alg: Madani} that depends on Madani's algorithm if the former cannot work properly.
 \begin{table}
 	\centering
 	\caption{Comparion of time complexity and measured running time in optimal control of the \textit{Ara} operon network}
 	\label{tbl: comparison}
 	\begin{adjustbox}{width=1\linewidth}
 		\begin{threeparttable}
 			\begin{tabular}{lcclc} 
 				\toprule
 				Method                        &  \cite{cheng2014optimal}                &  \cite{zhu2018optimal}               &Algorithm \ref{alg: value iteration}  & Algorithm \ref{alg: Madani}                  \\ 
 				\midrule
 				Time complexity               & $O(N^4)$ &   $ O(N^4)  $  &       \multicolumn{1}{c}{$ O(PMN) $\tnote{1}}         &  $ O(MN^2) $                     \\ 
 				\midrule
 				\multirow{3}{*}{Running time (s)} & \multirow{3}{*}{116736} & \multirow{3}{*}{57078} & 0.21 ($ \theta = 0.1 $)               & \multirow{3}{*}{7.64}  \\
 				&                         &                        & 0.28 ($ \theta = 0.01 $)                &                           \\
 				&                         &                        & 0.36 ($ \theta = 0.001 $)              &                           \\
 				\bottomrule
 			\end{tabular}
 			\begin{tablenotes}
 				\item[1] $ P $ refers to the number of iterations and is not bounded for an exact optimium.
 			\end{tablenotes}
 		\end{threeparttable}
 	\end{adjustbox}
 \end{table}
 
 \begin{remark}
 	The time complexity is stated to be $ O(MN + N^4) $ in \cite{zhu2018optimal}. We note that, in general, there exists $ M < N $ or even $ M \ll N $ in practice, i.e., fewer control inputs than state variables, especially for large networks \cite{lu2019pinning}. Besides, we can always assume $ M \le N $, since a state can transit to at most $ N $ succeeding states regardless of the number of control inputs, and it is useless to have more inputs than state variables. Thus, the time complexity of Algorithm \ref{alg: Madani} is equivalently $ O(N^3) $. Though the running time listed in Table \ref{tbl: comparison} may partly depend on implementation details, the difference in orders of magnitude demonstrates obviously the superiority of our approaches in terms of time efficiency.
 \end{remark}
 
 \section{Conclusions} \label{sec: conclusion}
 We tackled the infinite-horizon optimal control of BCNs with discounted cost in this paper. Unlike the existing methods, we solved this problem from the perspective of a deterministic Makov decision process (DMDP). We first showed that the control of a BCN could be well described by a DMDP and then proposed two approaches for the optimization of this DMDP, one based on value iteration and the other based on Madani's algorithm, while the latter can obtain the exact optimum with lower time complexity than existing work. Besides, the value iteration based approach can potentially get a near-optimal solution with much less running time than all other methods. A benchmark example using the \textit{ara} operon network has demonstrated the superior time efficiency of both proposed approaches. The DMDP view of BCN control may be promising for other problems as well and deserves more investigations. 






\bibliographystyle{IEEEtran}
\bibliography{boolnet}

\end{document}